\documentclass[11pt]{amsart}

\usepackage[english]{babel}
\usepackage{amsmath,amssymb,amsthm}
\usepackage[pdftex]{graphicx,color}

\newtheorem{theorem}{Theorem}[section]
\newtheorem{proposition}[theorem]{Proposition}

\newtheorem{corollary}[theorem]{Corollary}

\theoremstyle{definition}
\newtheorem{defn}[theorem]{Definition}
\newtheorem{example}[theorem]{Example}

\newcommand{\F}{\mathcal{F}}
\newcommand{\A}{\mathcal{A}}
\newcommand{\B}{\mathcal{B}}
\newcommand{\M}{\mathcal{M}}
\renewcommand{\H}{\mathcal{H}}
\newcommand{\K}{\mathcal K}      
\newcommand{\T}{\mathcal T}                    
\newcommand{\E}{\mathbb{E}}
\newcommand{\N}{\mathbb{N}}
\newcommand{\R}{\mathbb{R}}
\newcommand{\C}{\mathbb{C}}

\newcommand{\QE}[2]{\E_{#1}\left[ {#2}\right]}

\newcommand{\tr}{ \operatorname{Tr} }

\newcommand{\dd}{\mathrm{d}}
\renewcommand{\d}{\, \mathrm{d}}

\newcommand{\lvar}[2]{{{\rm Var}_\ell}_{#1}\left[ {#2}\right]} 
\newcommand{\rvar}[2]{{{\rm Var}_r}_{#1}\left[ {#2}\right]} 
\newcommand{\var}[2]{{\rm Var}_{#1}\left[ {#2}\right]}

\newcommand{\define}{\emph}

\newcommand{\Prob}[1]{\mathbf{P}\left\{{#1}\right\}}

\renewcommand{\Pr}{\mathbf{P}}

\newcommand{\essran}{\operatorname{ess-ran}}
\newcommand\cstar{{\rm C}^*}       
\newcommand\cstarconv{ {\rm C}^*{\rm conv}}
\newcommand\spec{\operatorname{Spec}}

\begin{document}

\title[Spectra and Variance of Quantum Random Variables]{Spectra and Variance of Quantum Random Variables}

\author{Douglas Farenick and Michael J.~Kozdron}
\address{Department of Mathematics and Statistics, University of Regina, Regina, Saskatchewan S4S 0A2, Canada}
\email{douglas.farenick@uregina.ca}
\email{kozdron@stat.math.uregina.ca}

\author{Sarah Plosker}
\address{Department of Mathematics and Computer Science, Brandon University, Brandon, Manitoba 
R7A 6A9, Canada} 
\email{ploskers@brandonu.ca}

\begin{abstract} 
We study essentially bounded quantum random variables and show that
the Gelfand spectrum of such a quantum random variable $\psi$ coincides with the hypoconvex hull of the
essential range of $\psi$. Moreover, a notion of operator-valued variance is introduced,
leading to a formulation of the moment problem in the context of quantum probability spaces in terms of operator-theoretic properties involving
semi-invariant subspaces and spectral theory. As an application of quantum variance, 
new measures of random and inherent quantum noise are introduced for measurements of quantum systems,
modifying some recent ideas of Polterovich \cite{polterovich2014}.
\end{abstract}



\maketitle

\section{Introduction}
\label{Introduction}

Some of the most basic and useful properties of classical random variables are altered when passing from
real- or complex-valued measurable functions to operator-valued measurable functions (that is, from classical to 
quantum random variables).  In earlier works~\cite{farenick--kozdron2012,farenick--plosker--smith2011,kylerthesis}, 
a certain operator-valued formulation of the notion of expectation of a quantum random variable was considered. 
In the present paper, we consider a similar formulation
for the variance of a quantum random variable. As in these earlier investigations, the noncommutativity of operator
algebra will lead to some structure that simply does not appear in the classical setting.

It is a basic fact of functional analysis that the essential range of an essentially bounded random variable 
coincides with the spectrum of a certain element in an abelian von Neumann algebra. Specifically, if 
$\psi:X\rightarrow\C$ is an essentially bounded function on a probability space $(X,\F(X), \mu)$, then the
essential range of $\psi$ is precisely the spectrum of $\psi$, where one considers $\psi$ as an element of the
von Neumann algebra $L^\infty(X,\F(X),\mu)$. We will arrive at a similar result for essentially bounded quantum random
variables on quantum probability spaces using higher dimensional spectra. However, it will turn out that our investigation 
of quantum variance will also involve notions from spectral theory. In particular, the quantum moment problem admits
a characterisation entirely within spectral terms.

As an application of our operator-valued variance, we consider some recent work of Polterovich~\cite{polterovich2014}
on random and inherent quantum noise in which the variance has a role. In Polterovich's work, a somewhat hybrid context
is at play: while the measures are operator-valued, the random variables are classical. In modifying Polterovich's ideas to
account for operator-valued measures and operator-valued random variables, we formulate new 
measures of quantum noise. One
of the main consequences of our results in this direction is that if an experimental apparatus is free of random
quantum noise, then it is classical, not quantum mechanical. Our work on quantum noise involves another idea 
that may be of value in other settings, namely that of quantum randomisation (or smearing), which is in contrast to 
the hybrid notion of smearing studied in early works such as~\cite{qm-book,jencova--pulmannov2009}. By way of quantum 
randomisation, 
we also modify another concept of Polterovich to obtain a measure of the 
intrinsic quantum noise of the apparatus represented by $\nu$.

If $(X,\F(X))$ denotes an arbitrary measurable space, and if $M$ is a von Neumann algebra with predual $M_*$ and positive cone $M_+$, then a 
function $\nu:\F(X)\rightarrow M$ is a positive operator-valued measure (POVM) if
\begin{enumerate}
\item $\nu(E) \in M_+$ for every $E \in \F(X)$,
\item $\nu(X) \neq 0$, and
\item $\omega\circ\nu:\F(X)\rightarrow\mathbb C$ is a complex measure for every $\omega\in M_*$.
\end{enumerate}
Note that the third condition above asserts that, 
for every countable collection $\{E_k\}_{k \in \N} \subseteq \F(X)$ with $E_j \cap E_k = \emptyset$ for $j \neq k$,
\begin{equation}\label{sum}
\nu\left(\bigcup_{k\in \N} E_k \right) = \sum_{k \in \N}\nu(E_k),
\end{equation}
where the convergence is with respect to the ultraweak topology of $M$.

If a POVM $\nu$ also satisfies $\nu(E\cap F)=\nu(E)\nu(F)$ for all $E,F\in\F(X)$, then $\nu$ is called a projective POVM.
An important theorem of M.A.~Naimark~\cite{naimark1943},~\cite[Theorem 4.6]{Paulsen-book} states that every POVM
admits a dilation to a projective POVM. Lastly, if a POVM $\nu$ has the property that $\nu(X)=1$, the identity element of $M$, then
$\nu$ is called a quantum probability measure.

A function $\psi:X\rightarrow M$ is said to be measurable if the complex-valued function $\omega\circ f$ on $X$
is measurable for every $\omega\in M_*$. Furthermore, if $\nu$ is a quantum probability measure, then a measurable 
function $\psi:X\rightarrow M$ is called a quantum random variable.

Suppose that $\omega\in M_*$ is a faithful state on $M$ and that $\nu$ is a quantum probability measure.  Then $\omega\circ\nu$ is a (classical)
probability measure and, because $\omega$ is faithful, $\nu$ and $\omega\circ\nu$ are mutually absolutely continuous.
The predual of the von Neumann algebra 
$L^\infty(X,\omega\circ\nu)\overline\otimes M$ is given by $L^1_{M_*}(X,\omega\circ\nu)$~\cite[Theorem~IV.7.17]{Takesaki-bookI}.
By way of this duality isomorphism, if
$\Psi\in L^\infty(X,\omega\circ\nu)\overline\otimes M$, then there is a bounded measurable function
$\psi:X\rightarrow M$ such that, for each $f \in L^1_{M_*}(X,\omega\circ\nu)$, the complex number
$\Psi(f)$ is given by
\begin{equation*}
\Psi(f)=\int_X \omega\left(f(x)\psi(x)\right) \d (\omega\circ\nu)(x).
\end{equation*}
Although $\psi$ is not unique, it is unique up to a set of $\omega\circ\nu$-measure zero.
We therefore identify $\Psi$ and $\psi$ and consider the elements 
of $L^\infty(X,\omega\circ\nu)\overline\otimes M$, in the case where $\nu$ is a quantum probability measure, 
to be bounded quantum random variables $\psi:X\rightarrow M$.

The general context described above for operator-valued measures and functions is considered in this paper 
only in the setting a finite factor $M$ of type ${\rm I}_d$; that is, $M=\B(\H)$ for some $d$-dimensional Hilbert space $\H$
and $d\in\mathbb N$. The predual $M_*$ is denoted by $\T(\H)$ (the Banach space of trace-class operators on $\H$).
Owing to the finite-dimensionality of $\H$, the Banach spaces $\B(\H)$ and $\T(\H)$ are equal as sets, but as Banach spaces
any one 
of these spaces is isometrically isomorphic to the dual of the other.
In this setting, the faithful normal state $\omega\in M_*$ is chosen to be the normalised trace and, for a fixed quantum
probability measure $\nu$, we denote by $\mu$ the classical probability measure 
$\mu=\frac{1}{d}\tr\circ \nu$, where $\tr$ is the canonical trace on $\B(\H)$. 
Because $\H$ has finite dimension $d$, 
we adopt the following notation:
\[
L_\H^\infty(X,\nu) = L^\infty(X,\mu)\overline\otimes\B(\H) \cong L^\infty(X,\mu)\otimes M_d(\C),
\]
where $M_d(\C)$ is the space of $d\times d$ matrices over $\C$.

(The restriction to factors of type ${\rm I}_d$ is made for two reasons. The first reason is that  
the notion of quantum measurement most
often in practice entails a POVM $\nu$ with values in $\B(\H)$ for a finite-dimensional Hilbert space $\H$.
The second reason is that certain results, when formulated for infinite-dimensional factors, become far less interesting than 
is the case with finite-dimensional factors. As an example of this particular situation, 
compare  \cite[Theorem 5.1]{farenick--plosker--smith2011}
on the affine structure of the set of all quantum probability measures (type ${\rm I}_d$ case) with the analogous result in \cite{Holevo-book}  (type ${\rm I}_\infty$ case). 
Specifically, in the type
${\rm I}_d$ case extremal quantum probability measures are certain linear combinations of point-mass measures, whereas in the type ${\rm I}_\infty$ case the set of projective
quantum probability measures is dense (with need of forming convex combinations of such projective measures) in the space of all quantum probability measures.)

Lastly, all homomorphisms and isomorphisms of C$^*$-algebras are assumed, without saying so 
each time, to be unital and $*$-preserving. If $Z$ is a compact Hausdorff space, then
$C(Z)$ is the unital abelian C$^*$-algebra of all continuous functions 
$f:Z\rightarrow\C$.

\section{Basic Properties of Measurability and Quantum Expectation}

Some elementary but useful facts concerning measurable functions are noted in this section.

\begin{theorem}\label{msblqrv} The following two statements are equivalent for a function
$\psi:X\rightarrow\B(\H)$.
\begin{enumerate}
\item $\psi$ is  measurable.
\item $\psi^{-1}(U)$ is a measurable set, for every open set $U\subseteq\B(\H)$.
\end{enumerate}
\end{theorem}

\begin{proof} Fix an orthonormal basis $\{e_1,\dots,e_d\}$ of $\H$. Because 
$\psi$ is measurable if and only if
each coordinate function $\psi_{ij}(x)=\langle\psi(x)e_j,e_i\rangle$ is 
measurable~\cite[Section~III]{farenick--plosker--smith2011}, and 
because $\B(\H)$ is topologically equivalent to $\C^{d^2}$ in the product topology, 
we may assume without loss of generality that
$\psi:X\rightarrow\C^{d^2}$ and $\psi(x)=\left(\psi_{11}(x),\psi_{12}(x),\dots,\psi_{dd}(x)\right)$ for $x\in X$. Furthermore, every open
set $U\subseteq\C^{d^2}$ will be viewed as a product $U=\displaystyle\prod_{i,j=1}^d U_{ij}$ of open sets $U_{ij}\subseteq\C$.
Suppose now that $\psi$ is measurable. As each $\psi_{ij}$ is therefore measurable, 
we have that $\psi_{ij}(U_{ij})\in\mathcal F(X)$
for every open set $U_{ij}\subseteq\C^{d^2}$. Thus, if $U=\displaystyle\prod_{i,j=1}^d U_{ij}$
is open in $\C^{d^2}$, then
\(
\displaystyle\psi^{-1}\left(U\right)\,=\,\bigcap_{i,j=1}^d\psi_{ij}^{-1}(U_{ij})
\)
is a measurable set.
Conversely, if $\psi^{-1}(U)$ is a measurable set, 
then for a fixed ordered pair $(k,\ell)$ and any open set $U_{k\ell}\subseteq\C$,
we have
\(
\psi^{-1}_{k\ell}(U_{k\ell})\,=\,\psi^{-1}(U),
\)
where $U=\displaystyle\prod_{i,j=1}^d U_{ij}$ is the open set for which $U_{ij}=\C$ 
for all $(i,j)\not=(k,\ell)$. Thus, 
$\psi_{k\ell}$ is a measurable function.
\end{proof}

Mimicking the classical definition of a regular probability measure, we have the following.

\begin{defn} A quantum
probability measure $\nu: \F(X) \to \B(\H)$ is \emph{regular} if for every $E \in \F(X)$, 
\begin{eqnarray*}
\nu(E)&=&\inf\{\nu(U)\,|\,U\subseteq X\mbox{ is open, and }E\subseteq U\}\\
&=&\sup\{\nu(K)\,|\,K\mbox{ is compact, and }K\subseteq E\}.
\end{eqnarray*}
\end{defn}

We note that because $\B(\H)$ is a von Neumann algebra, the infimum and supremum in the 
definition of regular measure above exist. Furthermore,
because the normalised trace $\tau$ on $\B(\H)$ is a normal linear functional, the induced 
classical probability measure $\mu=\tau\circ \nu$ on $(X,\F(X))$ is regular
if the quantum probability measure $\nu$ is. This leads to the  next result which is the quantum analogue of the classical Lusin theorem. 

\begin{theorem}\label{lusinthm}
If $\psi:X\rightarrow\B(\H)$ is a quantum random variable and if $\nu$ is a regular quantum 
probability measure on $(X,\F(X))$, where $\F(X)$ is the $\sigma$-algebra of 
Borel sets of a locally compact Hausdorff space $X$,
then for every $\varepsilon>0$ there is a continuous function 
$\vartheta:X\rightarrow\B(\H)$ with 
compact support such that $\mu(\{x\in X\,|\, \psi(x)\neq \vartheta(x)\})<\epsilon$. 
\end{theorem}

\begin{proof} Let $\{\psi_{ij}\}_{i,j=1}^d$ be the set of coordinate functions defined 
in the proof of Theorem~\ref{msblqrv}. Because $\nu$ is a regular measure and $\tau$ is a normal state, the 
induced measure $\mu=\tau\circ\nu$ is also regular.
Hence, the classical Lusin theorem may be invoked to obtain, for each $i$ and $j$, 
a continuous function $\vartheta_{ij}:X\rightarrow\C$ with compact support and such that 
$\mu(D_{ij})<\varepsilon/d^2$, where
\(
D_{ij}\,=\,\{x\in X\,|\,\psi_{ij}(x)\neq \vartheta_{ij} (x)\}
\).
Let $\vartheta:X\rightarrow\B(\H)$ be the continuous map induced by the coordinate functions $\vartheta_{ij}$, and 
define $D$ to be the set
\(
D\,=\,\{x\in X\,|\,\psi(x)\neq \vartheta (x)\}
\)
which is measurable by Theorem~\ref{msblqrv}. Because $\displaystyle D\subseteq \bigcup_{i,j=1}^d D_{ij}$, we 
deduce that 
$\displaystyle\mu(D)\leq \sum_{i,j=1}^d\mu(D_{ij})<\varepsilon$.
\end{proof}

The following theorem and two definitions summarise the results of~\cite[Section~III]{farenick--plosker--smith2011} 
relevant for our purposes; see also~\cite{farenick--kozdron2012, kylerthesis} for additional details.

\begin{theorem}
If $\nu$ is a quantum probability measure, then $\nu$ is absolutely continuous with respect to the induced classical 
measure $\mu$, and there exists a quantum random variable
denoted by $\displaystyle \frac{\dd\nu}{\dd\mu}$ such that 
\begin{equation}\label{rn defn}
\displaystyle\int_E\tr\left(\rho\,\displaystyle\frac{\dd\nu}{\dd\mu}(x) \right)\d\mu(x) \,=\,
\tr\left(\rho\,\nu(E)\right) \,,
\end{equation}
for all $E\in \F(X)$ and every density operator $\rho$.
\end{theorem}

The Borel function $\displaystyle\frac{\dd\nu}{\dd\mu}$  is
called the \emph{principal Radon-Nikod\'ym derivative of $\nu$} and is a positive operator for
$\mu$-almost all $x\in X$.

\begin{defn}\label{nuintdefn}
\begin{enumerate}
\item
A quantum random variable $\psi$ is \emph{$\nu$-integrable} if for every density operator $\rho$ the complex-valued function
\[
\psi_\rho(x)\,=\,\tr\left(\rho\,\left(\displaystyle\frac{\dd\nu}{\dd\mu}(x)\right)^{1/2}\psi(x)\left(\displaystyle\frac{\dd\nu}{\dd\mu}(x)\right)^{1/2}\right)\,,\;x\in X,
\]
is $\mu$-integrable.
\item
The \emph{integral} of a $\nu$-integrable function $\psi:X\rightarrow\B(\H)$
is defined to be the unique operator acting on $\H$ having the property that
\[
\tr\left(\rho\int_X\psi\d\nu\right)\,=\,\int_X\,\psi_\rho\d\mu\,,
\]
for every density operator $\rho$.
\end{enumerate}
\end{defn}

\begin{defn} If  
 $\nu:\F(X)\rightarrow\B(\H)$ is a quantum probability measure, then the map 
$\mathbb E_{\nu}:L_\H^\infty(X,\nu)\rightarrow\B(\H)$ defined by
\[
\QE{\nu}{\psi} = \int_X\psi\d\nu
\]
is called the \emph{quantum expectation} of $\psi$ with respect to $\nu$.
\end{defn}

A version of the following example first appeared in~\cite[Example~3.4]{farenick--plosker--smith2011}; see also~\cite[Theorem~2.3(4)]{farenick--kozdron2012}.

\begin{example}\label{quantumaverageexample} Let $X=\{x_1, \dots, x_n\}$ and let $\F(X)$ be the power set of $X$. 
If  $h_1, \dots, h_n\in \B(\H)_+$ are such that $h_1+\cdots + h_n=1\in \B(\H)$, and $\nu$ satisfies
$\nu(\{x_j\})=h_j$ for $j=1, \dots, n$, then for every $\psi:X\rightarrow \B(\H)$,
\[
\QE{\nu}{\psi}=\int_X\,\psi \d\nu=\sum_{j=1}^nh_j^{1/2}\psi(x_j)h_j^{1/2}.
\]
Thus one can view $\QE{\nu}{\psi}$ as a quantum averaging of $\psi$. 
\end{example}

Recall~\cite[Chapter~3]{Paulsen-book} that a linear map $\varphi:\A\rightarrow\B$ of unital C$^*$-algebras is a unital completely positive (ucp) 
map if $\varphi(1_\A)=1_\B$ and
the induced linear maps
\[
\varphi\otimes{\rm id_n}:\A\otimes M_n(\C)\rightarrow\B\otimes M_n(\C)
\]
are positive for every $n\in\mathbb N$.

\begin{theorem}\label{varineq} Quantum expectation is a completely positive operation. That is, the linear map 
$\mathbb E_{\nu}:L_\H^\infty(X,\nu)\rightarrow\B(\H)$
is a unital completely positive map, for every 
quantum probability measure $\nu$.
\end{theorem}

\begin{proof} 
The linearity of the map $\E_{\nu}$ follows readily by definition. 
Because
the algebra $L^\infty(X,\mu)$ is a unital abelian C$^*$-algebra, where $\mu$ is induced 
by a quantum probability measure $\nu$,
the Gelfand transform $\Gamma:L^\infty(X,\mu)\rightarrow C(Z_\nu)$ is a unital C$^*$-algebra 
isomorphism, where
$Z_\nu$ is the maximal ideal space of $L^\infty(X,\mu)$. The topological space $Z_\nu$ 
is necessarily compact, Hausdorff, and totally disconnected. 
Hence, if $\H$ has finite dimension $d$, then $L_\H^\infty(X,\nu)$ and $C(Z_\nu)\otimes\B(\H)$ are isomorphic C$^*$-algebras via the 
unital $*$-isomorphism 
\(
\Gamma_{\H}=\Gamma\otimes{\rm id}_{\B(\H)}: L_\H^\infty(X,\nu) \rightarrow C(Z_\nu)\otimes\B(\H).
\)
Because the map $\mathbb E_{\nu}\circ\Gamma_\H^{-1}:C(Z_\nu)\otimes\B(\H)\rightarrow\B(\H)$ is 
unital and completely positive~\cite[Theorem~3.5]{farenick--plosker--smith2011}
and because the homomorphism $\Gamma_\H^{-1}$ is completely positive, 
the linear map $\mathbb E_{\nu}:L_\H^\infty(X,\nu)\rightarrow\B(\H)$ is necessarily
completely positive.
\end{proof}

Theorem~\ref{varineq} gives rise to the following operator inequality.

\begin{corollary}[Schwarz Inequality] If the operators $h_1,\dots,h_n\in \B(\H)_+$ satisfy $h_1^2+ \cdots +h_n^2=1$, then for all $z_1,\dots,z_n\in\B(\H)$,
\begin{equation}\label{cs ineq}
\left( \sum_{j=1}^n h_j z_j h_j\right)^* \left( \sum_{j=1}^n h_j z_j h_j\right) \,\leq\, \sum_{j=1}^n h_jz_j^*z_j h_j.
\end{equation}

\end{corollary}

\begin{proof} Let $X=\{x_1,\dots,x_n\}$, and let $\F(X)$ be the power set of $X$. If $\nu$ is the quantum probability measure 
for which $\nu(\{x_j\})=h_j^2$, and if the quantum random variable $\psi:X\rightarrow \B(\H)$ is defined by $\psi(x_j)=z_j$, for each $j=1,\ldots,n$, then $\psi\in L_\H^\infty(X,\nu)$
and $\QE{\nu}{\psi}= h_1z_1h_1 + \cdots + h_nz_nh_n$ as in Example~\ref{quantumaverageexample}. By the Schwarz inequality for completely positive linear maps~\cite[Proposition~3.3]{Paulsen-book}, we have $\QE{\nu}{\psi}^*\QE{\nu}{\psi}\leq \QE{\nu}{\psi^*\psi}$, which is precisely inequality~\eqref{cs ineq}.
\end{proof}

\section{The Essential Range of Quantum Random Variables}

\begin{defn} 
Let $\psi:X\rightarrow\B(\H)$ be a quantum random variable. The \emph{essential range} of $\psi$
is the set $\essran \psi$ of all operators $\lambda\in\B(\H)$ for which
\(
\mu\left(\psi^{-1}(U)\right)\,>\,0,
\)
for every neighbourhood $U$ of $\lambda$.
\end{defn}

The essential range of $\psi:X\rightarrow\B(\H)$ is closed and the $\mu$-measure of the set 
$\{x\in X\,|\,\psi(x)\not\in \essran \psi\}$ is zero. Thus, if $\psi_1$ and $\psi_2$ determine the same element $\psi\in L_\H^\infty(X,\nu)$, 
then $\psi_1$ and $\psi_2$ have the same essential range. 
Our aim is to identify the essential range with certain spectral elements of~$\psi$.

\begin{defn} If $\A$ is a unital C$^*$-algebra and $a\in\A$, let $\cstar(a)$ be the unital C$^*$-subalgebra generated by $a$. For $d\in\mathbb N$, the set
\[
\spec^d(a)\,=\,\{\varrho(a)\,|\,\varrho:\cstar(a)\rightarrow M_d(\C) \mbox{ is a homomorphism}\}
\]
is called the \emph{Gelfand spectrum} of $a$.
\end{defn}

Of course, for many elements $a$, it will be the case that $\spec^d(a)$ is empty. A notable exception 
occurs with (essentially) bounded measurable functions $\psi:X\rightarrow\C$, 
in which case $\spec^1(\psi)$, where $\psi$ is considered to be an element of the abelian von 
Neumann algebra $L^\infty(X,\mu)$, coincides with the essential range of $\psi$~\cite{Lang-book}. 
However, the case of quantum random variables (see Theorem~\ref{reducingspectrum} below) requires the notion
of hypoconvexity~\cite[Definition~1.6]{salinas1979}.

\begin{defn} A nonempty compact subset $Q\subset \B(\H)$ is \emph{hypoconvex} if 
\begin{enumerate}
\item $u^*\lambda u\in Q$ for every unitary $u\in \B(\H)$ and $\lambda\in Q$, and
\item $\displaystyle\sum_{j=1}^m p_j\lambda_j\in Q$ for all $\lambda_1,\dots,\lambda_m\in Q$ and
projections $p_1,\dots,p_m\in \B(\H)$ satisfying $p_1+\cdots+p_m=1$ and $p_j\lambda_j=\lambda_jp_j$ for each $j$.
\end{enumerate}
If $Q\subset \B(\H)$ is an arbitrary nonempty compact set, then $Q^\sim$ denotes the \emph{hypoconvex hull}
of $Q$, namely, the smallest hypoconvex set that contains $Q$.
\end{defn}

If $\H=\C$, then the two conditions above for the hypoconvexity of a compact set
$Q$ are trivially satisfied. Hence, the notion of hypoconvex set is distinguished from compactness only at
dimension $d=2$ and higher.

\begin{theorem}\label{reducingspectrum} The following two statements are equivalent for 
a Hilbert space $\H$ of dimension $d$ and a quantum random variable
$\psi\in L_\H^\infty(X,\nu)$.
\begin{enumerate}
\item $\lambda\in \spec^d(\psi)$.
\item There exists a unitary  $v:\C^d\rightarrow \H$ such that 
$v\lambda v^{-1}\in\left(\essran \psi\right)^\sim$.
\end{enumerate}
\end{theorem}

\begin{proof} Consider the isomorphism 
$\Gamma_\H:L_\H^\infty(X,\nu)\rightarrow C(Z_\nu)\otimes\B(\H)$
defined in the proof of Theorem~\ref{varineq}, and fix a computational basis $\mathfrak C=\{e_1,\dots,e_d\}$ of $\H$.
Let $\pi_{\mathfrak C}:\B(\H)\rightarrow M_d(\C)$ be the isomorphism that sends each rank-1 operator
$e_i\otimes e_j\in\B(\H)$ to the canonical matrix unit $E_{ij}$ of $M_d(\C)$. Thus, 
$\rho_{\mathfrak C}={\rm id}_{C(Z_\nu)}\otimes\pi_{\mathfrak C}$ is an isomorphism of $C(Z_\nu)\otimes\B(\H)$
and $C(Z_\nu)\otimes M_d(\C)$ that maps each $f\in C(Z_\nu)\otimes\B(\H)$ to a $d\times d$ matrix 
$\rho_{\mathfrak C}(f)=[f_{ij}]_{i,j}$ of continuous maps $f_{ij}:Z_\nu\rightarrow\C$. 

Suppose now that $\psi\in L_\H^\infty(X,\nu)$ and let 
$f=\Gamma_\H(\psi)$ and $[f_{ij}]_{i,j}=\rho_{\mathfrak C}(f)
=\rho_{\mathfrak C}\circ\Gamma_\H(\psi)$. Define a subset $\Omega_{\psi, \mathfrak C}\subset M_d(\C)$ by setting
\[
\Omega_{\psi, \mathfrak C}\,=\,\left\{ [\omega(f_{ij})]_{i,j=1}^d\in M_d(\C)\,|\,\omega:\cstar\left(\{f_{ij}\}_{i,j}\right)\rightarrow\C
\mbox{ is a homomorphism}\right\}.
\]
Note that by restricting the domain of a homomorphism $\omega:C(Z_\nu)\rightarrow\C$
to the unital C$^*$-subalgebra $\cstar\left(\{f_{ij}\}_{i,j}\right)$ of $C(Z_\nu)$ we obtain an inclusion
$\Delta_{\psi, \mathfrak C}\subseteq \Omega_{\psi, \mathfrak C}$, where 
\[
\Delta_{\psi, \mathfrak C}\,=\,\left\{ [\omega(f_{ij})]_{i,j=1}^d\in M_d(\C)\,|\,\omega:C(Z_\nu)\rightarrow\C
\mbox{ is a homomorphism}\right\}.
\] 
On the other hand, because every homomorphism $\omega_0:\cstar\left(\{f_{ij}\}_{i,j}\right)\rightarrow\C$
is, by the fact that $\C$ is $1$-dimensional, irreducible, there is a homomorphism 
$\omega:C(Z_\nu)\rightarrow\C$ such that 
$\omega_{\vert \cstar\left(\{f_{ij}\}_{i,j}\right)}=\omega_0$. 
Hence, $\Delta_{\psi, \mathfrak C}= \Omega_{\psi, \mathfrak C}$.
Because homomorphisms $C(Z_\nu)\rightarrow\C$ are point evaluations $g\mapsto g(z_0)$ for 
$z_0\in Z_\nu$, the set $\Omega_{\psi, \mathfrak C}$ is the range of the matrix-valued function 
$z\mapsto[f_{ij}(z)]_{i,j=1}^d$, which we have denoted by $\rho_{\mathfrak C}\circ\Gamma_\H(\psi)$.
Thus, $\Omega_{\psi, \mathfrak C}^\sim=\spec^d\left(\rho_{\mathfrak C}\circ\Gamma_\H(\psi)\right)$~\cite[Lemma 2.3]{salinas1979}.
The Gelfand spectrum is an isomorphism invariant; hence, $\spec^d(\psi)=\Omega_{\psi, \mathfrak C}^\sim$.

Now let $v_{\mathfrak C}:\C^d\rightarrow\H$ be the unitary operator that 
sends the $j$-th coordindate vector of $\C^d$ to the unit vector $e_j\in \H$ in the 
computational basis $\mathfrak C=\{e_1,\dots,e_d\}$ of $\H$. If 
$\lambda=[\lambda_{ij}]_{i,j}\in \Omega_{\psi, \mathfrak C}$, then
there is a $z_0\in Z_\nu$ such that $\lambda_{ij}=f_{ij}(z_0)=\langle f(z_0)e_j,e_i\rangle$. Hence, 
$v_{\mathfrak C}\lambda v_{\mathfrak C}^{-1}$ is an element of the range of $f=\Gamma_\H(\nu)$, 
which coincides with the essential
range of $\psi$. As the map $s\mapsto v_{\mathfrak C}sv_{\mathfrak C}^{-1}$ is an automorphism of $\B(\H)$, 
it is also true that if $\lambda\in\Omega_{\psi, \mathfrak C}^\sim$, then 
$v_{\mathfrak C}\lambda v_{\mathfrak C}^{-1}$ is an element of the hypoconvex hull of the range of $f$.
This completes the proof that 
$\lambda\in \spec^d(\psi)$ only if there exists a unitary operator $v:\C^d\rightarrow \H$ such that 
$v\lambda v^{-1}\in\left(\essran\psi\right)^\sim$.

Conversely, 
for each choice of computational basis $\mathfrak C=\{e_1,\dots,e_d\}$ of $\H$ there is an isometry
$v_{\mathfrak C}:\C^d
\rightarrow\H$ that sends the $j$-th coordindate vector of $\C^d$ to the unit vector $e_j\in \H$.
Hence, if $v_{\mathfrak C}\lambda v_{\mathfrak C}^{-1}\in\left(\essran \psi\right)^\sim$, then
$\lambda\in \spec^d(\psi)$.
\end{proof}

The expectation $\QE{\nu}{\psi}$ of $\psi$ is just one of many operators $\varphi(\psi)\in\B(\H)$
obtained by evaluating $\psi$ at a ucp map $\varphi:L^\infty_\H(X,\nu)\rightarrow  \B(\H)$; that is,
\[
\QE{\nu}{\psi}\in \{\Phi(\psi)\,|\, \Phi:L^\infty_\H(X,\nu)\rightarrow  \B(\H)\textnormal{ is a ucp map}\}. 
\]
Theorem~\ref{sarah} below clarifies the relationship between operators of this type and the essential range of $\psi$.

\begin{theorem}\label{sarah} If $\varphi:L^\infty_\H(X,\nu)\rightarrow  \B(\H)$ is a ucp map, then for every
$\psi\in L^\infty_\H(X,\nu)$ there exist 
$x_1,\dots,x_m\in X$ (not necessarily distinct) and $t_1,\dots, t_m\in\B(\H)$ such that
\[
\varphi(\psi)\,=\,\sum_{j=1}^m t_j^*\psi(x_j)t_j \quad\mbox{ and }\quad
\sum_{j=1}^mt_j^*t_j=1\in\B(\H).
\]
\end{theorem}

\begin{proof} Because $L^\infty_\H(X,\nu)$ is isomorphic to 
$L^\infty(X,\mu)\otimes M_d(\C)$, every $\psi\in L^\infty_\H(X,\nu)$ may be represented as a 
$d\times d$ matrix whose entries are taken from an abelian C$^*$-algebra; that is,
$\psi$ is $d$-normal. Select an orthonormal basis $\{e_1,\dots,e_d\}$ of $\H$ and
consider the $d\times d$ matrix $\omega=\left[ \langle\varphi(\psi)e_j,e_i\rangle\right]_{i,j=1}^d$.
Thus, $\omega$ is an element of the matricial range $W^d(\psi)$, namely the set of all
$d\times d$ matrices of the form $\Phi(\psi)$ for some ucp map 
$\Phi:L^\infty_\H(X,\nu)\rightarrow M_d(\C)$. By~\cite[Theorem~3.9]{bunce--salinas1976},
there are $\lambda_1,\dots,\lambda_q\in\spec^d(\psi)$   
(not necessarily distinct) and $s_1,\dots, s_q\in M_d(\C)$ such that $\omega=\sum_j s_j^*\lambda_j s_j$ and
$\sum_js_j^*s_j=1$. By Theorem~\ref{reducingspectrum},
there exist unitaries $v_j:\C^d\rightarrow \H$ such that 
$v_j\lambda_j v_j^*\in\left(\essran \psi\right)^\sim$. For each $j=1,\dots, q$
there are $x_1^{(j)},\dots,x_{n_j}^{(j)}\in X$ and pairwise-orthogonal projections $p_1^{(j)},\dots,p_{n_j}^{(j)}\in \B(\H)$
such that each $p_i^{(j)}$ commutes with $\psi(x_i^{(j)})$ and $\displaystyle v_j\lambda_j v_j^{*}=\sum_{i=1}^{n_j}\psi(x_i^{(j)})p_i^{(j)}$.
Let $u:\C^d\rightarrow\H$ be the unitary for which $u^*zu=[\langle ze_j,e_i\rangle]_{i,j}$ for all $z\in\B(\H)$. 
In particular, $\omega=u^*\varphi(\psi)u$ and each $s_j=u^*r_ju$ for a unique $r_j\in\B(\H)$. Thus,
\begin{align*}
\varphi(\psi)=\displaystyle\sum_{j=1}^q\sum_{i=1}^{n_j} us_j^*\lambda_j s_ju^* 
&= \displaystyle\sum_{j=1}^q\sum_{i=1}^{n_j} (p_i^{(j)}u^*r_j)^* \psi(x_i^{(j)}) (p_i^{(j)}u^*r_j) \\
&= \displaystyle\sum_{\ell=1}^m t_\ell^*\psi(x_\ell)t_\ell,
\end{align*}
where $\{t_\ell\}_\ell=\{ p_i^{(j)}u^*r_j\}_{i,j}$ and $\{x_\ell\}_\ell=\{x_{i}^{(j)}\}_{i,j}$ are relabelings and
renumberings of the operators and points in the decomposition of $\varphi(\psi)$ above.
\end{proof}

Recall that a subset $K\subseteq\B(\H)$ is \emph{C$^*$-convex} if 
\(
\displaystyle \sum_{j=1}^m t_j^*z_jt_j\in K\,,
\)
for every $z_1,\dots, z_m\in K$ and $t_1,\dots, t_m\in\B(\H)$ with $\displaystyle \sum_{j=1}^mt_j^*t_j=1\in\B(\H)$.
If $S\subset\B(\H)$ is a nonempty set, then $\cstarconv(S)$ is the smallest C$^*$-convex set 
that contains $S$. Thus, Theorem~\ref{sarah} leads immediately to the following corollary which is a generalisation of~\cite[Theorem~2.3(8)]{farenick--kozdron2012}.

\begin{corollary}\label{sarah-cor} If $\psi\in L^\infty_\H(X,\nu)$, then
\[
\{\Phi(\psi)\,|\, \Phi:L^\infty_\H(X,\nu)\rightarrow  \B(\H)\textnormal{ is a ucp map}\}
=\cstarconv(\essran \psi).
\]
In particular, $\QE{\nu}{\psi}\in \cstarconv(\essran \psi)$.
\end{corollary}

\section{Quantum Variance}

 \begin{defn}\label{variance} If $\psi \in L^\infty_\H(X,\nu)$ is a quantum random variable, then
\begin{enumerate}
\item the \define{left variance} of $\psi$ with respect to $\nu$
is the operator 
\[
\lvar{\nu}{\psi}=   \QE{\nu}{\psi^*\psi}- \QE{\nu}{\psi^*}\QE{\nu}{\psi},
\]
\item the \define{right variance} of $\psi$ with respect to $\nu$
is the operator 
\[
\rvar{\nu}{\psi}=   \QE{\nu}{\psi\psi^*}- \QE{\nu}{\psi}\QE{\nu}{\psi^*},
\]
and
\item the \define{variance} of $\psi$ with respect to $\nu$
is the operator 
\[
\var{\nu}{\psi}=\frac{1}{2}\left( {\lvar{\nu}{\psi}} +  {\rvar{\nu}{\psi}}  \right).
\]
\end{enumerate}
\end{defn}

The Schwarz inequality ensures that all three of the variances defined above are positive operators. However, this occurrence of positivity
is a consequence of the fact that $\psi$ is essentially bounded and that $L_\H^{\infty}(X,\nu)$ is a von Neumann algebra.
In contrast, variance is defined classically for square-integrable random variables rather than essentially bounded random variables, as it
is a result of Chebyshev's inequality that square-integrable random variables are necessarily integrable (i.e., $L^2 \subset L^1$).
To similarly define the variance of an arbitrary quantum random variable $\psi$, it is necessary to fulfil the \emph{second moment condition} 
that  $\psi^*\psi$ be $\nu$-integrable.  The obvious question is whether or not the second moment condition implies that $\psi$  is itself $\nu$-integrable;
the theorem below answers this question affirmatively. 

\begin{theorem}\label{cheb}
Suppose that $\psi: X \to \B(\H)$ is a quantum random variable. If $\psi^*\psi$ is $\nu$-integrable, then $\psi$ is $\nu$-integrable.  
\end{theorem}

\begin{proof}
Let $\rho$ be a density operator and consider the functions $(\psi^*\psi)_\rho$ and $\psi_\rho$ on $X$ defined by 
\[
\begin{array}{rcl}
(\psi^*\psi)_\rho(x)&=&\tr\left(\rho^{1/2}\,\left(\displaystyle\frac{\dd\nu}{\dd\mu}(x)\right)^{1/2}[\psi^*\psi(x)]\left(\displaystyle\frac{\dd\nu}{\dd\mu}(x)\right)^{1/2}\rho^{1/2}\right),\textnormal{ and}\\ && \\
\psi_\rho(x)&=&\tr\left(\rho^{1/2}\,\left(\displaystyle\frac{\dd\nu}{\dd\mu}(x)\right)^{1/2}\psi(x)\left(\displaystyle\frac{\dd\nu}{\dd\mu}(x)\right)^{1/2}\rho^{1/2}\right),
\end{array}
\]
which coincide with Definition~\ref{nuintdefn} using elementary properties of the trace.
Let us also define, using the constant function $\iota(x)=1\in\B(\H)$, the scalar-valued function 
\[
\iota_\rho(x)\,=\,\tr\left(\rho\,\displaystyle\frac{\dd\nu}{\dd\mu}(x)\right).
\]
Note that $\iota_\rho$ is $\mu$-integrable. To complete the proof we shall require the following two tracial inequalities for arbitrary $y,z\in \B(\H)$:
\begin{enumerate}
\item\label{g} (\cite[Theorem~1]{gardner1979}) $|\tr(y)|\leq\tr(|y|)$, and
\item\label{a} (\cite[Theorem~1]{bhatiakittaneh}) $\displaystyle\tr(|yz^*|)\leq\frac{1}{2}\tr(|y|^2)+ \frac{1}{2}\tr(|z|^2)$.
\end{enumerate}
Consider the function $w:X\rightarrow\B(\H)$ defined by
$w(x)=\displaystyle\left(\frac{\dd\nu}{\dd\mu}(x)\right)^{1/2}\rho^{1/2}$. 
The two tracial inequalities above imply that
\begin{align*}
|&\psi_\rho(x)|=\left|\tr\left(w(x)^*\psi(x)w(x)\right)\right|  
= \left|\tr\left([\psi(x)w(x)]w^*(x)\right)\right|\\
&\leq  \tr\left(\left|[\psi(x)w(x)]w^*(x)\right|\right) 
\leq \frac{1}{2}\tr(|\psi(x)w(x)|^2) +\frac{1}{2}\tr(|w(x)|^2) \\
&= \frac{1}{2} \tr(w(x)^*\psi(x)^*\psi(x)w(x)) +\frac{1}{2}\tr(w^*(x)w(x))
= \frac{1}{2}\psi^*\psi_\rho(x) + \frac{1}{2}\iota_\rho (x).
\end{align*}
Thus, $|\psi_\rho|$ is bounded above by the average of the two nonnegative $\mu$-integrable functions
$(\psi^*\psi)_\rho$ and $\iota_\rho$. Hence, $\psi_\rho\in L^1(X,\mu)$. As this is true for every density operator $\rho$, 
we deduce that $\psi$ is $\nu$-integrable.
\end{proof}

\begin{corollary}\label{extended var} The three variances in Definition~\ref{variance} can be defined for
quantum random variables $\psi$ for which $\psi^*\psi$ is $\nu$-integrable
\end{corollary} 

Notwithstanding the extension of the variance domains as indicated in Corollary~\ref{extended var}, it is not necessarily true
that the left or right variance is positive. In other words, there is no natural analogue of the Schwarz inequality from essentially
bounded quantum random variables to square-integrable quantum random variables.

We now turn our attention to essentially bounded
quantum random variables having variance zero. Although random variables having variance zero 
are trivially constant in the classical case, we will show that a much richer structure exists for quantum random 
variables having variance zero. 

One family of quantum random variables that have variance zero is the following. For $z\in\B(\H)$, let $\psi_z:X\rightarrow\B(\H)$ 
denote the constant function defined by $\psi_z(x)=z$ for every $x\in X$. Because
of the noncommutativity of operator algebra, quantum averaging of $z$, by way of $z\mapsto \QE{\nu}{\psi_z}$, 
may in fact alter $z$.  This phenomenon was observed in~\cite[Theorem 2.3(8)]{farenick--kozdron2012}, where it was shown that, 
in general, one only has $\QE{\nu}{\psi_z} \in \cstarconv(\{z\})$.  However, if $\QE{\nu}{\psi_z}=z$, 
namely if quantum averaging does not disturb $z$, then the variance of $\psi_z$ is zero; this is the 
immediate analogue of the fact that scalars (i.e., constant random variables) have variance zero.

\begin{proposition} If $ \QE{\nu}{\psi_z}=z$, then $\var{\nu}{\psi_z}=0$.
\end{proposition}

\begin{proof} 
By~\cite[Theorem 2.3(8)]{farenick--kozdron2012}, the set of all $y\in\B(\H)$ for which $\QE{\nu}{\psi_y}=y$ is a unital C$^*$-subalgebra
of $\B(\H)$. Hence, $\QE{\nu}{\psi_z}=z$ implies that ${\lvar{\nu}{\psi}}={\rvar{\nu}{\psi}}=0$.
\end{proof}

The following result is a concise spectral characterisation of variance zero in the case of essentially bounded quantum random variables.

\begin{theorem}\label{varzerothm}
The following two statements are equivalent for the quantum random variable $\psi\in L_\H^\infty(X,\nu)$.
\begin{enumerate}
\item $\var{\nu}{\psi}=0$.
\item There exist a unitary  $u:\H\rightarrow\C^d$ and a $\lambda\in \mbox{\rm Spec}^d(\psi)$ such that
$u^*\lambda u=\QE{\nu}{\psi}$.
\end{enumerate}
\end{theorem}

\begin{proof} The condition that $\var{\nu}{\psi}=0$ is equivalent to the two equations 
$\QE{\nu}{\psi^*\psi}=\QE{\nu}{\psi}^*\QE{\nu}{\psi}$ and $\QE{\nu}{\psi\psi^*}=\QE{\nu}{\psi}\QE{\nu}{\psi}^*$ holding 
simultaneously, which in turn is equivalent to $\psi$ belonging to the multiplicative domain of the ucp map
$\mathbb E_\nu$~\cite[Theorem~3.18]{Paulsen-book}. Because the multiplicative domain of $\mathbb E_\nu$ is a unital
C$^*$-subalgebra of $L_\H^\infty(X,\nu)$ and contains $\psi$, the restriction of $\mathbb E_\nu$ to $\cstar(\psi)$
is a homomorphism. Thus, by selecting an orthonormal basis $\{\phi_1,\dots,\phi_d\}$ of $\H$ and in letting
$u:\H\rightarrow\C^d$ be the unitary operator that sends each $\phi_j$ to $e_j$, we have that $\lambda=u^* \QE{\nu}{\psi} u$
is an element of $\mbox{\rm Spec}^d(\psi)$. Conversely, if there exist $\lambda\in \mbox{\rm Spec}^d(\psi)$ and 
a unitary $u:\H\rightarrow\C^d$ such that $u^*\lambda u=\QE{\nu}{\psi}$, then the restriction of $\mathbb E_\nu$ to $\cstar(\psi)$
is a homomorphism and so $\var{\nu}{\psi}=0$.
\end{proof}

\section{The Quantum Moment Problem}

The classical Hamburger moment problem, named after the German mathematician Hans Ludwig Hamburger, is as follows. 
Suppose that $\{g_k\}_{k\in\mathbb N}$ is a sequence of real numbers.  Does there exist a positive Borel measure 
$\mu$ on the real line such that
\begin{equation}\label{solvable}
g_k = \int_{-\infty}^{\infty} x^k \d \mu \; ?
\end{equation}
This problem, as well as many variants of it, has been extensively studied for almost a century.  
If $\{g_k\}_{k\in\mathbb N}$ is given, then we say that $\mu$ is a \define{solution} to the moment problem for  
$\{g_k\}_{k\in\mathbb N}$ if~\eqref{solvable} holds.

A condition for a unique solution to one variant of the moment problem that is very well known 
in classical probability is found in~\cite[Theorem~30.1]{Billingsley}, 
namely if $g=\{g_k\}_{k\in\mathbb N}$ is given and satisfies
\begin{equation}\label{billchar}
P_g(t)=\sum_{k=0}^\infty \frac{g_k t^k}{k!} <\infty
\end{equation} 
for all $t \in \R$, then there is a  unique probability measure $\Pr$ on $\R$ that is the solution to the 
moment problem for  $\{g_k\}_{k\in\mathbb N}$. 
In the case that $\{g_k\}_{k\in\mathbb N}$ is a \define{multiplicative moment sequence},
meaning that $g_k = (g_1)^k$ for all $k$ for some $g_1 \in \R$, the unique solution 
to the moment problem is trivial. That is,~\eqref{billchar} implies the solution to the moment problem is unique since
\(
P_g(t) = e^{g_1 t} <\infty
\)
for all $t \in \R$. If $\Pr$ denotes the probability measure 
on $\R$ supported on $g_1$ and $Y$ is a random variable with  $\Prob{Y = g_1} = 1$,
then $\int_\R Y^k \d \Pr  = (g_1)^k$.
Thus,  
only constant random variables have 
multiplicative moment sequences.

Suppose now that $\{g_k\}_{k\in\mathbb N}$ is a sequence in $\B(\H)$. We say that a
quantum probability measure $\nu$ on the Borel sets of $\B(\H)$ is a solution to the 
quantum moment problem for $\{g_k\}_{k\in\mathbb N}$ if 
there exists a Borel subset $X \subseteq \B(\H)$ and a quantum random variable $\psi: X \to \B(\H)$ such that
$$\QE{\nu}{\psi^k} =\int_X\psi^k \d \nu= g_k$$
for all $k = 0,1,2,\ldots$. 
A natural question to ask is if we can develop an operator-theoretic 
criterion to determine when the quantum moment problem for a multiplicative moment sequence has only a trivial solution.

By Stinespring's dilation theorem for unital completely positive linear maps~\cite[Theorem~4.1]{Paulsen-book}, we deduce that
for every quantum probability measure $\nu$, there exist a
Hilbert space $\K_\nu$, 
an isometry $v:\H\rightarrow\K_\nu$, and
a homomorphism $\Delta_\nu:L_\H^\infty(X,\nu)\rightarrow\B(\K_\nu)$
such that
\begin{enumerate}
\item $\QE{\nu}{\psi}=v^*\Delta_\nu(\psi)v$ for every $\psi\in L_\H^\infty(X,\nu)$, and
\item $\mbox{Span}\,\{\Delta_\nu(\psi)v\xi\,|\,\psi\in L_\H^\infty(X,\nu),\,\xi\in\H\}$ is dense in $\K_\nu$.
\end{enumerate}
The two conditions above determine the triple $(\K_\nu, \Delta_\nu, v)$ up to unitary equivalence~\cite[Proposition~4.2]{Paulsen-book}, and so we
may refer unambiguously to the triple $(\K_\nu, \Delta_\nu, v)$ as the minimal
Stinespring dilation of the quantum expectation $\mathbb E_\nu$. (Here, ``minimal'' is in reference to the second condition, which is to say that
the Hilbert space $\K_\nu$ is no larger than it needs to be.)

A second operator-theoretic concept that we will employ is that of a semi-invariant subspace.
 A subspace $\M$ of a Hilbert space $\K$ is said to be semi-invariant for an operator $z\in\B(\K)$ 
if $\M=\mathcal L_0^\bot\cap\mathcal L_1$ for some $z$-invariant 
subspaces $\mathcal L_0$ and $\mathcal L_1$. 

 \begin{theorem}\label{trivial moments} Let $\nu$ be a quantum probability measure, and let $\psi\in L_\H^\infty(X,\nu)$.
 Assume that $\{g_k\}_{k\in\mathbb N}\subset\B(\H)$ is a sequence of operators with $g_k=\QE{\nu}{\psi^k}$, and let
$(\K_\nu, \Delta_\nu, v)$ be a minimal Stinespring dilation of the quantum expectation $\mathbb E_\nu$.
Then the following two statements are equivalent.
\begin{enumerate}
\item $g_k=(g_1)^k$ for every $k\in\mathbb N$.
\item $v(\H)$ is a semi-invariant subspace for $\Delta_\nu(\psi)$.
\end{enumerate} 
\end{theorem}

\begin{proof} By a result of Sarason~\cite[Lemma 0]{sarason1965}, a subspace $\M$ of $\K_\nu$ is semi-invariant for $\Delta_\nu(\psi)$
if and only if $p\Delta_\nu(\psi)^k{}_{\vert \M}=\left( p\Delta_\nu(\psi)_{\vert \M}\right)^k$ for every $k\in\mathbb N$,  
where $p\in\B(\K_\nu)$ is the projection with range $\M$. In the case at hand, the dimension of $\M$ is necessarily $d$. Further, any
projection $p\in \B(\K_\nu)$ of rank $d$ can be factored as $p=vv^*$ for some isometry $v:\H\rightarrow\H_\nu$ and, conversely,
for every isometry $v:\H\rightarrow\H_\nu$ the operator $vv^*$ is a projection of rank $d$. Because $\Delta_\nu$ is a homomorphism,
Sarason's criterion is, for an isometry $v:\H\rightarrow\H_\nu$, equivalent to: $v(\H)$ is semi-invariant for $\Delta_\nu(\psi)$
if and only if $v^*\Delta_\nu(\psi^k)v=\left(v^*\Delta_\nu(\psi)v\right)^k$ for every $k\in\mathbb N$. Thus, because $g_k=v^*\Delta_\nu(\psi^k)v$
for all $k\in\mathbb N$, the proof of the theorem is complete.
\end{proof}

In order to state our final result, we recall the following definition~\cite[Definition~5.1]{bunce--salinas1976}.

\begin{defn} Assume that $\A$ is a unital C$^*$-algebra and that $k\in\mathbb N$.
The $k\times k$ \define{matricial spectrum} of $a\in\A$ is the subset $\sigma^k(a)\subset M_k(\C)$ defined by
\[
\sigma^k(a)\,=\,\left\{\varphi(a)\,|\,\varphi:\A\rightarrow M_k(\C)\mbox{ is ucp and }\varphi_{\vert\mathcal R(a)}
\mbox{ is a homomorphism}\right\},
\]
where $\mathcal R(a)$ is the rational Banach subalgebra of $\A$ generated by $a$.
\end{defn}

To be more precise, the algebra $\mathcal R(a)$ is the norm-closure of the abelian algebra of all elements of the form $p(a)q(a)^{-1}$, where
$p$ and $q$ are complex polynomials such that $q$ has no roots in the spectrum $\sigma(a)$ of $a$.

If one considers the classical $d=1$ case, then every point $\lambda$ in the spectrum of $\psi$ gives rise  to a measure $\mu$ for which
$\lambda^k=\int_X\psi^k\d\mu$ for every $k\in\mathbb N$. Indeed, this measure $\mu$ is a point-mass measure corresponding to a point evaluation of 
$\psi$ that yields the complex number $\lambda$ (which is basically the situation described at the end of the second paragraph of this section). 
The matter is simplified somewhat by the fact that $\sigma^1(\psi)=\mbox{Spec}^1(\psi)$
if $\psi$ is a classical random variable. However, in higher dimensions, the matricial
spectrum $\sigma^d(\psi)$ is generally much larger than the Gelfand
spectrum $\mbox{Spec}^d(\psi)$ and, consequently, the quantum moment problem for multiplicative moment sequences entails certain obstructions not seen at the classical level. 
Our final result, Theorem~\ref{voiculescu}, illustrates the obstruction in that it demonstrates 
that a correction by a unitary quantum random variable is necessary prior to integration. The underlying complicating factor is
that an analogue of the Riesz Representation Theorem holds only for a certain subset of ucp maps $L_\H^\infty(X,\nu)$~\cite[Corollary 4.5]{farenick--plosker--smith2011}.

If $(X,\F(X))$ is the Borel space of a compact metric space $X$, and if $\psi:X\rightarrow M_d(\C)$ is
continuous, then below we consider $\psi$ as an element of the unital C$^*$-algebra $C(X)\otimes M_d(\C)$, and the
matricial spectra of $\psi$ are defined relative to this choice of C$^*$-algebra.

\begin{theorem}\label{voiculescu} If $\psi:X\rightarrow M_d(\C)$ is a continuous quantum random variable on a compact metric space $X$, 
and if $\lambda\in\sigma^d(\psi)$, then there exist sequences $\{\nu_n\}_{n\in\mathbb N}$ and $\{w_n\}_{n\in\mathbb N}$ of quantum probability measures and
quantum random variables, respectively, such that
\begin{enumerate}
\item $w_n(x)$ is unitary for all $x\in X$ and every $n\in\mathbb N$, and
\item $\displaystyle\lim_{n\rightarrow\infty}\left\|\lambda^k-\QE{\nu_n}{w_n^*\psi^kw_n}\right\|=0$ for every $k\in\mathbb N$.
\end{enumerate}
\end{theorem}

\begin{proof} By hypothesis there is a ucp map $\vartheta:C(X)\otimes M_d(\C)\rightarrow M_d(\C)$ such that the restriction of
$\vartheta$ to the rational algebra $\mathcal R(\psi)$ is a homomorphism. Let $v^*\Delta v$ be a minimal Stinespring representation
of $\vartheta$.
Because $X$ is a metric space, the C$^*$-algebra $C(X)\otimes M_d(\C)$ is separable; hence, the minimal Stinespring
dilation $\Delta$ of $\vartheta$ takes place on a representing Hilbert space $\H_\Delta$ that is separable. Hence, by Voiculescu's Theorem~\cite[Corollary II.5.9]{Davidson-book}, $\Delta$ is approximately unitarily equivalent to a direct sum $\tilde\Delta$ of a countable family
of irreducible representations of $C(X)\otimes M_d(\C)$. Because  $C(X)\otimes M_d(\C)$ is homogenous,
every irreducible representation of it is a point evaluation. Thus, there are a countable subset $X_1\subseteq X$, 
a separable Hilbert space $\H_{\tilde\Delta}=\displaystyle\bigoplus_{x\in X_1}\C^d_x$, where $\C^d_x=\C^d$ for each $x\in X_1$, and a sequence $\{u_n\}_{n\in\mathbb N}$
of unitary operators $u_n:\H_{\Delta}\rightarrow\H_{\tilde\Delta}$ such that, for every $f\in C(X)\otimes M_d(\C)$,
\[
\tilde\Delta(f)=\bigoplus_{x\in X_1}f(x) \quad\mbox{and}\quad\lim_{n\rightarrow\infty}\|\Delta(f)-u_n^*\tilde\Delta(f)u_n\|=0.
\]
Hence,
\[
\lim_{n\rightarrow\infty}\|\vartheta(f)-(u_nv)^*\tilde\Delta(f)(u_nv)\|=0
\] 
for every $f\in C(X)\otimes M_d(\C)$.

The isometry $u_nv:\C^d\rightarrow\displaystyle\bigoplus_{x\in X_1}\C^d_x$
acts as $u_nv\xi=\displaystyle\bigoplus_{x\in X_1}v_{n,x}\xi$ for $\xi\in\C^d$, for some $v_{n,x}\in M_d(\C)$ and has the property that
$\sum_{x\in X_1}v_{n,x}^*v_{n,x}=1$ for each $n$. By the Polar Decomposition, there is a unitary operator $w_{n,x}\in\B(\C^d_x)=M_d(\C)$ such that $v_{n,x}=w_{n,x}|v_{n,x}|$. 
Now define $w_n:X\rightarrow M_d(\C)$
by $w_n(x)=w_{n,x}$, if $x\in X_1$, and $w_n(x)=1$ if $x\not\in X_1$. Because point sets are closed in a Hausdorff space, each $w_n$ is a measurable function.
By defining $\nu_n:\F(X)\rightarrow M_d(\C)$ by 
\[
\nu_n=\sum_{x\in X_1}\delta_{\{x\}}v_{n,x}^*v_{n,x},
\]
where $\delta_{\{x\}}$ is a classical point-mass measure concentrated at $\{x\}$, we see that $\nu_n$ is a quantum
probability measure such that
\(
(u_nv)^*\tilde\Delta(f)(u_nv)=\QE{\nu_n}{w_n^*fw_n}
\)
for every $f\in C(X)\otimes M_d(\C)$. Thus, using the fact that $\vartheta$ is a homomorphism on $\mathcal R(\psi)$, we have that
\[
\lambda^k=\vartheta(\psi)^k=\vartheta(\psi^k)=\lim_{n\rightarrow\infty}\QE{\nu_n}{w_n^*\psi^kw_n}
\]
for every $k\in\mathbb N$.
\end{proof}
 
In the following special case, one can dispense with the sequences $\{\nu_n\}_n$ and $\{w_n\}_n$, and the quantum moment problem for multiplicative moment sequences is solved exactly rather than asymptotically.

\begin{corollary}\label{the end} If $\psi:X\rightarrow M_d(\C)$ is a quantum random variable on $X=\{x_1,\dots,x_n\}$ and if
$\lambda\in\sigma^d(\psi)$, then there exist a quantum probability measure $\nu$ on $(X,\F(X))$, where $\F(X)$ is the power set of $X$, and a unitary-valued quantum random variable
$w:X\rightarrow M_d(\C)$ such that
$\lambda^k=\QE{\nu}{w^*\psi^kw}$ for every $k\in\mathbb N$.
\end{corollary}

\begin{proof}  The C$^*$-algebra $C(X)\otimes M_d(\C)$ has, in this case, finite dimension. Therefore, the minimal Stinesping dilation of every ucp map on
$C(X)\otimes M_d(\C)$ has a representing Hilbert space of finite dimension. Hence, in the proof of Theorem~\ref{voiculescu}, the representations $\Delta$
and $\tilde\Delta$ may be assumed to be equal.
\end{proof}

\section{Measures of Quantum Noise}

Our focus to this point has been with purely mathematical issues. However, 
the probability measures that we have studied herein feature
prominently in the theory of quantum measurement. In this regard, the variance
of a quantum random variable has a particularly crucial role.

To explain briefly the physical context, assume that a $d$-dimensional
Hilbert $\H$ is used to model (the states of) some physical quantum system. The states of the
quantum system are represented by density operators $\rho$ acting on $\H$. The system
will have various physical properties; those properties of the system that can actually be measured 
using some experimental apparatus or device are called observable properties. In the mathematical
formulation of quantum theory, an observable property is represented by a hermitian operator, while an
experimental apparatus is represented by a quantum probability measure $\nu$ on $(X,\F(X))$,
where $X$ is the sample space of possible outcomes of the measurement and $\F(X)$ is
a $\sigma$-algebra of events. Therefore, in practice, $X$ is a finite set and $\F(X)$ is the power set of $X$.
Our assumptions here about $X$ are a little more general: namely, that $X$ is a compact Hausdorff space and
that $\F(X)$ contains the Borel sets of $X$. The statistical element of quantum measurement is realised 
by the following axiom: if, at the moment of the measurement, the system is in state $\rho$, then the
probability that event $E\in\F(X)$ will be measured is $\tr\left(\rho\nu(E)\right)$.

The observable properties of a system associated with a particular quantum measurement $\nu$ will, in general, 
intermingle information about the system with random disturbances coming from the measuring apparatus.
These random disturbances are called quantum noise of $\nu$. (The physics of
quantum noise is treated in \cite{qn}, for example.)
To quantify the amount of quantum noise present in
a quantum mechanical measurement, various numerical measures of quantum noise have been introduced
(see, for example, \cite{bhl2004,polterovich2012,polterovich2014}).
Two forms of 
quantum noise---random noise and inherent noise---have been investigated 
recently by Polterovich \cite{polterovich2014}.
In Polterovich's approach, the ``quantum'' aspect is 
captured by a certain scalar-valued measurable function
(specifically, a Markov kernel), which is integrated with respect to a POVM to 
produce a Hilbert space operator whose norm is used
to determine a numerical indicator of the amount of quantum noise present. 
We outline below how a similar process is carried out using
operator-valued measurable functions; we adopt, as much as possible, the notation of Polterovich. 

Suppose that $\nu$ and $\H$ are fixed, and consider $K(\nu)$, the closed unit ball of $L^\infty_{\H}(X,\nu)$.
The \emph{random quantum noise} of $\nu$ is the quantity $N(\nu)$ defined by
\[
N(\nu)=\sup_{\psi\in K(\nu)}\| \var{\nu}{\psi} \|.
\]
Because quantum expectation is a contractive completely positive map, the operator inequality
\[
\lvar{\nu}{\psi}=   \QE{\nu}{\psi^*\psi}- \QE{\nu}{\psi^*}\QE{\nu}{\psi} \leq \QE{\nu}{\psi^*\psi}
\]
yields the norm inequality
\[
\left\| \lvar{\nu}{\psi} \right\| \leq \left\| \QE{\nu}{\psi^*\psi} \right\|
\leq \| {\psi^*\psi}\| =\|\psi\|^2.
\]
Likewise, $\left\| \rvar{\nu}{\psi} \right\| \leq \|\psi\|^2$ for all $\psi$. Hence,
\[
0\leq \sup_{\psi\in K(\nu)}\| \var{\nu}{\psi} \| \leq 1.
\]

Our definition above of the random quantum noise of $\nu$ differs from that of Polterovich (see \cite[p.~489]{polterovich2014}), although we
have used the same notation. The difference lies in the fact that we are using a larger class of functions $\psi$ in defining $K(\nu)$---that is,
we use operator-valued $\psi$, not just scalar-valued $\psi$.

\begin{proposition}\label{no noise} $N(\nu)=0$ if and only if the mass of $\nu$ is concentrated at a point $x_0$ of $X$.
\end{proposition}

\begin{proof} If $N(\nu)=0$, then $\var{\nu}{\psi}=0$
for every $\psi\in K(\nu)$. Hence, by Theorem \ref{varzerothm}, 
the quantum expectation map $\mathbb E_\nu$ is a unital homomorphism of $L^\infty(X,\mu)\otimes M_d(\mathbb C)$ onto
$M_d(\mathbb C)$. If $Z$ is the maximal ideal space of $L^\infty(X,\mu)$, then the unital homomorphisms of the homogenous
C$^*$-algebra $C(Z)\otimes M_d(\mathbb C)$ onto $M_d(\mathbb C)$ are point evaluations $f\mapsto f(x_0)$. Hence, there is
an $x_0\in X$ such that $\nu=\delta_{\{x_0\}}1$.

Conversely, if $\nu$ has its mass concentrated at a point, then $\mathbb E_\nu$ is a unital homomorphism and so
$\var{\nu}{\psi}=0$
for every $\psi\in K(\nu)$.
\end{proof}

Returning to the postulate that quantum probability measures are associated with measurements of quantum systems, 
Proposition \ref{no noise} has the following consequence.

\begin{corollary} Every apparatus that performs measurements of a physical quantum system admits random quantum noise.
\end{corollary}

\begin{proof} Suppose that $\nu$ is the quantum probability measure associated with the measurement apparatus of some physical quantum system
represented by a finite-dimensional Hilbert space.
If $N(\nu)=0$, then $\nu=\delta_{\{x_0\}}1$ for some $x_0\in X$; in other words, the probability is exactly $1$
that an event $E$ containing outcome $x_0$ is measured, regardless of the state
of the system. However, this contravenes the axioms of quantum mechanics. Hence, it must be that $N(\nu)>0$.
\end{proof} 

A subtler and potentially more descriptive notion of quantum noise is that of \emph{inherent quantum noise}. To discuss inherent quantum noise, we 
first extend the concept of ``smearing or randomisation of a measurement'' 
\cite{qm-book,jencova--pulmannov2009,polterovich2014} to one
which involves quantum random variables rather than classical random variables. 

Assume that $(X,\F(X))$ and $(Y,\F(Y))$ are Borel spaces for compact Hausdorff spaces $X$ and $Y$, and suppose
that $\nu$ and $\nu'$ are quantum probability measures on $(X,\F(X))$ and $(Y,\F(Y)$, respectively, with values in $\B(\H)$ for some
$d$-dimensional Hilbert space $\H$. The measure $\nu$ is said to be a 
\emph{quantum randomisation} of $\nu'$ if there exists a function 
$\gamma$ (sending $y$ to $\gamma_y$) 
of $Y$ into the space of quantum probability measures on $(X,\F(X))$ with values in $\B(\H)$ such that
\begin{enumerate}
\item for every $E\in \F(X)$, the map $f_E^\gamma:Y\rightarrow \B(\H)$ defined by $f_E^\gamma(y)=\gamma_y(E)$, for $y\in Y$, is a
measurable function on $(Y,\F(Y))$, and
\item $\nu(E)=\displaystyle\int_Y f_E^\gamma\, \d\nu'$, for every $E\in \F(X)$.
\end{enumerate}
Furthermore, the linear transformation $\Gamma_{\nu'}:L^\infty_\H(X,\nu)\rightarrow L^\infty_\H(Y,\nu')$ defined by
\[
\Gamma_{\nu'}\psi(y)=\int_X\psi\,\d\gamma_y,
\]
is called a \emph{quantum randomisation operator}.

\begin{proposition} The quantum randomisation operator
$\Gamma_{\nu'}$ is a unital completely positive linear map.
\end{proposition}

\begin{proof} The linearity of $\Gamma_{\nu'}$ has already been noted, and it is clear that $\Gamma_{\nu'}(1)=1$ because each POVM $\gamma_y$
satisfies $\gamma_y(X)=1$. To show that $\Gamma_{\nu'}$ is completely positive, let $n\in\mathbb N$ be given and suppose that
$\Psi=[\psi_{ij}]_{i,j=1}^n\in M_n\left( L^\infty_\H(X,\nu)\right)$ is positive. 
Because $L^\infty_\H(X,\nu)$ is a homogenous C$^*$-algebra, 
so is $M_n\left( L^\infty_\H(X,\nu)\right)$. Indeed, 
\[
M_n\left( L^\infty_\H(X,\nu)\right) \simeq C(Z)\otimes M_{nd}(\mathbb C),
\]
where $Z$ is the maximal ideal space of the abelian von Neumann algebra $ L^\infty(X,\mu)$. Thus, to say that the
matrix $\Psi$ is positive in $M_n\left( L^\infty_\H(X,\nu)\right)$ is to say that the operator matrix $\Psi(x)=[\psi_{ij}(x)]_{i,j=1}^n$ 
acting on $\displaystyle\bigoplus_1^n \H$
is a positive operator for $\mu$-almost all $x\in X$. Likewise, 
$\Gamma_{\nu'}^{(n)}(\Psi)=\left[\Gamma_{\nu'}(\psi_{ij})\right]_{i,j=1}^n$ is positive in 
$M_n\left( L^\infty_\H(Y,\nu')\right)$ if $\Gamma_{\nu'}^{(n)}(\Psi)(y)$ is a positive operator matrix for 
$\mu'$-almost all $y\in Y$. Now if $y\in Y$, then
\[
\Gamma_{\nu'}^{(n)}(\Psi)(y)=\left[\Gamma_{\nu'}(\psi_{ij})(y)\right]_{i,j=1}^n
=\Gamma_{\nu'}^{(n)}(\Psi)=\left[\QE{\gamma_y}{\psi_{ij}}\right]_{i,j=1}^n
=\mathbb E_{\gamma_y}^{(n)}[\Psi].
\]
Because $\gamma_y$ is a quantum probability measure, the expectation $\mathbb E_{\gamma_y}$ is completely positive. Thus, 
$\mathbb E_{\gamma_y}^{(n)}[\Psi]=\Gamma_{\nu'}^{(n)}(\Psi)(y)$ is a 
positive operator on $\displaystyle\bigoplus_1^n \H$, which proves 
that $\Gamma_{\nu'}$ is completely positive, and completes the proof.
\end{proof}

Another useful property of the quantum randomisation operator $\Gamma_{\nu'}$ is the following.

\begin{proposition} $\mathbb E_{\nu'}\circ\Gamma_{\nu'}=\mathbb E_\nu$.
\end{proposition}

\begin{proof} Select $E\in \F(X)$ and consider the quantum random 
variable $\chi_E$ (the characteristic function of $E$). If $y\in Y$, then
\[
\Gamma_{\nu'}\chi_E(y)=\int_X\chi_E\,\d\gamma_y=\gamma_y(E)=f_E^\gamma(y).
\]
Hence,
\[
\QE{\nu'}{\Gamma_{\nu'}\chi_E}=\int_Y f_E^\gamma\,\d\nu'=\nu(E)=\QE{\nu}{\chi_E}.
\]
Because the span of the characteristic functions is norm dense in $L^\infty_H(X,\nu)$, the linearity and continuity of
$\Gamma$ and of the expectations $\mathbb E_{\nu'}$ and $\mathbb E_{\nu}$ yield 
$\QE{\nu'}{\Gamma_{\nu'}\psi}=\QE{\nu}{\psi}$,
for every $\psi\in L^\infty_\H(X,\nu)$.
Thus, 
$\mathbb E_{\nu'}\circ\Gamma_{\nu'}=\mathbb E_\nu$
as required.
\end{proof}

The main result of this section is the following theorem,
which states that to every quantum random variable $\psi$ on $(X,\F(X),\nu)$
there corresponds a quantum random variable $\Gamma_{\nu'}(\psi)$ on $(Y,\F(Y),\nu')$
such that $\Gamma_{\nu'}(\psi)$ and $\psi$ have the same quantum expectation and the
quantum variance of $\Gamma_{\nu'}(\psi)$ is bounded above in the Loewner ordering of $\B(\H)_{\rm sa}$
by the quantum variance of $\psi$.

\begin{theorem}\label{iqn1} $\var{\nu'}{\Gamma_{\nu'}\psi}\leq  \var{\nu}{\psi}$.
\end{theorem} 

\begin{proof} Consider first the left variance. Given $\psi\in L^\infty_\H(X,\nu)$, we have that
$\Gamma_{\nu'}(\psi^*\psi)\geq\Gamma_{\nu'}(\psi)^*\Gamma_{\nu'}(\psi)$ because every unital completely positive linear map
satisfies the Schwarz inequality. Hence, 
\[
\begin{array}{rcl} 
 \lvar{\nu}{\psi} &=&    \QE{\nu}{\psi^*\psi}-\QE{\nu}{\psi}^*\QE{\nu}{\psi} \\ && \\
 &=&  \QE{\nu'}{\Gamma_{\nu'}(\psi^*\psi)}-\QE{\nu'}{\Gamma_{\nu'}\psi}^*\QE{\nu'}{\Gamma_{\nu'}\psi} \\ && \\
 &\geq &\QE{\nu'}{\Gamma_{\nu'}(\psi)^*\Gamma_{\nu'}(\psi)}-
 \QE{\nu'}{\Gamma_{\nu'}\psi}^*\QE{\nu'}{\Gamma_{\nu'}\psi} \\ && \\
 &=& \lvar{\nu'}{\Gamma_{\nu'}\psi}.
\end{array}
\]
A similar inequality holds for the right variance. 
Therefore, the inequality holds for the average of the left and right variances; hence, 
$\var{\nu'}{\Gamma_{\nu'}\psi}\leq  \var{\nu}{\psi}$.
\end{proof}

The \emph{intrinsic quantum noise} of a quantum probability measure $\nu$ on $(X,\F(X))$ is the 
quantity $N_{\rm in}(\nu)$ defined by
\[
N_{\rm in}(\nu)=\inf_{(\Gamma_{\nu'},\nu')}\sup_{\psi\in K(\nu)}\| \var{\nu'}{\Gamma_{\nu'}\psi} \|.
\]

An immediate consequence of Theorem \ref{iqn1} is a fundamental inequality that extends 
a similar inequality of Polterovich \cite[Proposition 2.2]{polterovich2014}.

\begin{theorem}\label{iqn2} $0\leq N_{\rm in}(\nu)\leq N(\nu)\leq 1$. 
\end{theorem}

\section*{Acknowledgements}

The work of all three authors is supported, in part, by the Natural Sciences and Engineering Research Council of Canada. 
The second author thanks the Isaac Newton Institute for Mathematical Sciences, Cambridge, for its hospitality during the 
Random Geometry programme in Spring 2015 where part of this paper was completed.


\end{document}